\documentclass{amsart}
\usepackage{graphicx}
\usepackage{caption}
\usepackage{subcaption}
\usepackage{mathtools}
\usepackage{amsmath}
\usepackage{amsthm}
\usepackage{amssymb}
\usepackage[normalem]{ulem}
\usepackage[pdftex,  colorlinks=true]{hyperref}
\usepackage{accents}
\usepackage{enumitem}
\usepackage{algorithm}
\usepackage{algpseudocode}
\usepackage{comment}
\usepackage{nicefrac}
\usepackage{xcolor}
\usepackage{cleveref}
\usepackage{thmtools}
\usepackage{thm-restate}

\usepackage{tikz-cd}

\newtheorem{theorem}{Theorem}
\newtheorem{proposition}[theorem]{Proposition} 

\newtheorem{corollary}[theorem]{Corollary} 
 
\newtheorem{lemma}[theorem]{Lemma}
\newtheorem{definition}[theorem]{Definition}
\newtheorem{remark}[theorem]{Remark}

\newcommand{\dist}{\mathsf{dist}}

\DeclareMathOperator{\APBP}{APBP}
\DeclareMathOperator{\GTREE}{GromovTree}
\DeclareMathOperator{\GPRD}
{ToProd}
\DeclareMathOperator{\LPRD}{LocProd}

\DeclareMathOperator{\EXTE}{METR}

\DeclareMathOperator{\IGPRD}
{ToDist}

\begin{document}

\title[Gromov's Tree and the APBP Problem]{ Gromov's Approximating Tree and the All-Pairs Bottleneck Paths Problem}

\author[ ]{Anders Cornect} 
\address{University of Waterloo, Waterloo, ON, Canada}
\email{acornect@uwaterloo.ca}

\author[ ]{Eduardo Mart\'inez-Pedroza }  
\address{Memorial University of Newfoundland, St. John's, NL, Canada}
\email{emartinezped@mun.ca}

\date{\today}

\begin{abstract}
Given a pointed metric space $(X,\mathsf{dist}, w)$ on $n$ points, Gromov's approximating tree is a 0-hyperbolic pseudo-metric space $(X,\mathsf{dist}_T)$ such that $\mathsf{dist}(x,w)=\mathsf{dist}_T(x,w)$ and $\mathsf{dist}(x, y)-2 \delta \log_2n \leq \mathsf{dist}_T (x, y) \leq \mathsf{dist}(x, y)$ for all $x, y \in X$ where $\delta$ is the Gromov hyperbolicity constant of $X$. On the other hand, the all pairs bottleneck paths (APBP) problem asks, given an undirected graph with some capacities on its edges, to find the maximal path capacity between each pair of vertices. In this note, we prove:    \begin{itemize}
        \item Computing Gromov's approximating tree for a metric space with $n+1$ points from its matrix of distances reduces to solving the APBP problem on an connected graph with $n$ vertices.
        \item There is an explicit algorithm that computes Gromov's approximating tree for a graph from its adjacency matrix in quadratic time.
        \item Solving the APBP problem on a weighted graph with $n$ vertices reduces to finding Gromov's approximating tree for a metric space with $n+1$ points from its distance matrix.
    \end{itemize}
\end{abstract}

\maketitle


\section*{Introduction}

The notion of $\delta$-hyperbolic space was introduced by Gromov. A pseudo-metric space $(X,\dist)$ is $\delta$-hyperbolic if it satisfies that, for any $x,y,z,w \in X$,
\begin{equation} \nonumber
\dist(x,z) + \dist(y,w) \leq   \max\{\dist(x,y)+\dist(z,w), \dist(y,z)+\dist(x,w) \} + 2\delta.
\end{equation}

Any finite pseudo-metric space is $\delta$-hyperbolic for some $\delta\geq0$ and $\delta$ can be regarded as a measure of tree-likeness. This can be justified by the following facts:
\begin{enumerate}
    \item The path metric on every weighted tree (a simplicial tree on which each edge has been assigned a positive length) is $0$-hyperbolic; see~\cite[Page 322]{Dr84} and references therein.
    \item A metric space $(X,\dist)$ isometrically embeds into a weighted tree if and only if $(X,\dist)$ is $0$-hyperbolic; see~\cite[Theorem 8]{Dr84} or references therein. 
    \item (Gromov's approximating tree) For any $\delta$-hyperbolic metric space $(X,\dist)$ and any $w\in X$ there is a $0$-hyperbolic metric space $(X,\dist_T)$ such that $\dist_T(w,x)=\dist(w,x)$ and 
    $\dist(x,y)-2\delta \log_2 |X| \leq \dist_T(x,y) \leq \dist(x,y)$; see~\cite[\S 6.1]{Gr87} or~\cite[\S 2]{GH88}. 
\end{enumerate}

The set of pseudo-metrics on a set $X$ is partially ordered, $\dist_1 \leq \dist_2$ if and only if $\dist_1(x,y)\leq \dist(x,y)$ for any $x,y\in X$. We show that Gromov's approximating tree is optimal with respect to this order in following sense.

\begin{proposition}[Proposition~\ref{prop:Optimal}]\label{prop:OptimalIntro}
Gromov's approximating tree of $(X,\dist)$ with respect to $w\in X$ is the maximum of the set of all $0$-hyperbolic pseudo-metrics on $X$ that do not increase distances with respect to $\dist$ and respects distances to $w$.      
\end{proposition} 

A pseudo-metric space $(X,\dist)$ with $n$ points and a fixed enumeration $X=\{x_1,\ldots ,x_n\}$ is determined by its \emph{distance matrix} 
$D=(d_{ij})$ where $d_{ij}=\dist(x_i,x_j)$. We also show that Gromov's argument proving the existence of Gromov's approximating tree proves the following statement.

\begin{proposition}[Proposition~\ref{prop:7}]\label{prop:01-intro}
The distance matrix of Gromov's approximating tree of a metric space on $n$ points can be computed from the distance matrix of the space with time complexity $O(n^2)$. 
\end{proposition}

The quadratic time complexity in the above statement arises from interpreting Gromov's  definition in~\cite{Gr87} of Gromov's approximating tree for a metric space on $n$ vertices as the solution of the all pairs bottleneck path problem (APBP problem) for a complete graph on $n$ vertices with capacities on its edges, see Proposition~\ref{prop:7}.  The APBP problem in an undirected graph on $n$ vertices with real capacities can be computed in time $O(n^2)$, see~\cite{Hu61}. Conversely, we show that an algorithm that computes the distance matrix of Gromov's approximating tree from the distance matrix of the metric space can be used to solve the APBP problem, see Proposition~\ref{prop:11}. Hence, we deduce the following statement.

\begin{proposition}[Corollary~\ref{cor:14}]\label{prop:APBPreductionIntro}
   Solving the APBP problem on an undirected graph on $n$ vertices with positive capacities is equivalent to finding Gromov's approximating tree of a metric space on $n+1$ points.    
\end{proposition}

By a \emph{graph metric space} $(X,\dist)$ we mean a metric space where distances take integer values and such that there is a connected simple graph $\Gamma$ with vertex set $X$ such that $\dist(x,y)$ is the length of the shortest edge-path between vertices; let us emphasize that all edges of $\Gamma$ are assumed to have length one. In this case we say that $\Gamma$ \emph{realizes} $(X,\dist)$, and the Gromov's approximating tree of $(X,\dist)$ is also referred to as the Gromov's approximating tree of the connected graph $\Gamma$.

Distance approximating trees of connected graphs are a rich field of study within computational and applied graph theory. There are particular constructions of distance approximating trees of graphs that can be computed in linear time on the number of vertices of the graph, but they are either not optimal in the sense of Proposition~\ref{prop:OptimalIntro} or they increase distances, see for example~\cite{FRT2008, CD2000, DY2006}. 

Gromov's approximating tree of a connected graph is harder to compute directly from its adjacency matrix  in the following sense. Suppose that $\Gamma$ is a connected graph with $n$ vertices. Let $G$ denote the  adjacency matrix of $\Gamma$ and let $D$ be the distance matrix of the vertex set $\Gamma$. The   work of Seidel~\cite{SE95} shows that one can compute $D$ from $G$ in time complexity $O(n^\omega \log n)$, where $O(n^\omega)$ is the complexity of multiplying two $n\times n$ matrices. The best known upper bound of $\omega$ is $2.371552$, see the recent article~\cite{williams2023new}. Hence, based on Proposition~\ref{prop:01-intro}, there is an algorithm that computes the distance matrix of Gromov's approximating tree of $\Gamma$ from $G$ in time complexity $O(n^\omega \log n)$. This is not the best procedure, as we can show the following statement.

\begin{proposition}\label{prop:GraphsIntro}
There is an algorithm that computes Gromov's approximating tree of a connected graph in $n$ vertices from its adjacency matrix in time complexity $O(n^2)$.    
\end{proposition}

Through the article, we work in the class of pseudo-metric spaces. All graphs considered are simple graphs, that is, they are unoriented, 
 no edge connects a vertex to itself, and no two vertices are connected by more than one edge. 
 
 The rest of this article is organized into six sections. In the first section, we discuss the notion of Gromov product and introduce some mathematical notation that is used for the remainder of the note. In the second section, we discuss Gromov's approximating tree and prove that it is an optimal approximation in the sense of Proposition~\ref{prop:OptimalIntro}. The third section explains the relation between computing Gromov's approximating tree and the APBP problem; in particular, it includes the proof of Proposition~\ref{prop:01-intro}. The fourth section discusses the proof of Proposition~\ref{prop:APBPreductionIntro}, that is, how the APBP problem can be reduced to computing a Gromov's approximating tree. Finally, the fifth section is on Gromov's approximating trees of graphs and the proof of Proposition~\ref{prop:GraphsIntro}. 
 
\subsection*{Acknowledgements} 
The authors thank the referee of the article for a careful reading and tremendous feedback that improved the quality of this work. The first author acknowledges funding by the Natural Sciences and Engineering Research Council of Canada NSERC, via Undergraduate Student Research Awards (USRA) during 2023 and 2024. The second author acknowledges funding by the Natural Sciences and Engineering Research Council of Canada NSERC.

\section{Metric spaces and Gromov products}

Let $(X, \dist)$ be a finite pseudo-metric space on $n$ points with a fixed enumeration
\begin{equation}\label{eq:01} X=\{x_1,\ldots ,x_n\}.\end{equation}
For any $x,y,w\in X$, the Gromov product of $x$ and $y$ at $w$ is defined as
\begin{equation}\label{eq:02} (x|y)_w = \frac12\left( \dist(x,w)+\dist(y,w)-\dist(x,y) \right) .\end{equation} 
We regard the  pseudo-metric space $(X,\dist)$ as its $n\times n$ \emph{distance matrix} $D$, 
\begin{equation}\label{eq:03}
D=(d_{ij}),\qquad d_{ij}=\dist(x_i,x_j).
\end{equation}
Fix a point $w\in X$. Assume, without loss of generality, that
\begin{equation}\label{eq:04}
w=x_n.\end{equation}
Let $L$ be the $n\times n$ matrix of Gromov products in $(X,\dist)$ at $w$, that is, 
\begin{equation} \label{eq:05} L=(\ell_{ij}), \qquad \ell_{ij}=(x_i|x_j)_w.\end{equation}
The definition of Gromov product~\eqref{eq:02} and the assumption $w=x_n$ imply that  \begin{equation}\label{eq:relations} d_{ni}=\ell_{ii},\quad  d_{ij}= \ell_{ii}+\ell_{jj} -2\ell_{ij}, \quad 2\ell_{ij}= d_{ni}+d_{nj}-d_{ij} .\end{equation}
Hence, there are algorithms that compute $L$ from $D$, and  $D$ from $L$, both with time complexity $O(n^2)$. In a diagram,
\begin{equation}
 \begin{tikzcd}
   D\arrow[rr, bend left, "O(n^2)"] && L \arrow[ll, bend left, "O(n^2)"] 
\end{tikzcd}   
\end{equation}
Let us name these two procedures which play an essential role in the statement of our main results.

\begin{definition}[$\mathcal{D}_n$ and $\mathcal{P}_n$]
    Let $\mathcal{D}_n$ be the set of $n\times n$  matrices  which are distance matrices of pseudo-metric spaces with $n$ points. Let $\mathcal{P}_n$  
denote the set of $n\times n$  matrices   which are matrices of Gromov products  of pointed pseudo-metric spaces with $n$ points.
\end{definition}

\begin{definition}[$\GPRD$ and $\IGPRD$]\label{def:1}
Let the functions \[\GPRD\colon \mathcal{D}_n \to \mathcal{P}_n\qquad \text{and} \qquad \IGPRD\colon \mathcal{P}_n \to \mathcal{D}_n\] 
be defined 
as follows. 
If $D=(d_{ij})$ is the distance matrix of a pseudo-metric space with $n$ points as in~\eqref{eq:03}, let $L=\GPRD(D)$ be the matrix $L=(\ell_{ij})$ of Gromov products given by~\eqref{eq:relations}. 
If $L=(\ell_{ij})$ is the matrix of Gromov products of a metric space with $n$ points as in~\eqref{eq:05}, let $D=\IGPRD(L)$ be the distance matrix given by~\eqref{eq:relations}. 
\end{definition}

\begin{remark}\label{rem:Inverses}
As functions, $\GPRD$ and $\IGPRD$ are inverses, 
\begin{equation}\label{eq:08}
    \IGPRD \circ \GPRD = \mathsf{Id}_{\mathcal{D}_n}  \quad \text{and} \quad \GPRD\circ\IGPRD = \mathsf{Id}_{\mathcal{P}_n},
\end{equation}
and moreover, they both can be computed in  time complexity $O(n^2)$ over $\mathcal{D}_n$ and $\mathcal{P}_n$ respectively.  
\end{remark}

\section{Gromov's Approximating tree}

Let $\delta\geq0$. The pseudo-metric space $(X,\dist)$ is \emph{$\delta$-hyperbolic} if 
\[ (x|z)_w \geq \min\{ (x|y)_w, (y|z)_w \} -\delta    \]
for every $x,y,z,w\in X$. 
The constant  \[\delta_* = \max_{x,y,z,w\in X}\Big\{  \min\{ (x|y)_w, (y|z)_w \} -(x|z)_w \Big\}\]
is called the \emph{Gromov hyperbolicity constant} of $(X,\dist)$. It is an observation that this definition is equivalent to the one in the introduction. A \emph{weighted tree} is a metric space whose underlying space is a simplicial tree where each edge has been assigned a positive length, and whose metric is the induced length metric. As we remarked in the introduction, a weighted tree is a $0$-hyperbolic space.
The following result by Gromov illustrates that the hyperbolicity of a finite pseudo-metric space is a measure of tree-likeness.  

\begin{theorem}[Gromov's Approximating Tree] \cite[Page 155]{Gr87}\label{thm:GromovTree}
Let $(X,\dist)$ be a $\delta$-hyperbolic pseudo-metric space with $n$ points. For any $w\in X$ there exists a weighted tree $T$ and a surjective map $\varphi\colon X \to T$ such that 
\[ \dist(x, w) = \dist_T (\varphi(x), \varphi(w))\]  and
\[ \dist(x, y)-2 \delta \log_2n   \leq \dist_T (\varphi(x), \varphi(y)) \leq \dist(x, y)\] 
for all $x, y \in X$.  
\end{theorem}
The pair $(T,\varphi)$ given by the theorem is called \emph{Gromov's approximating tree of $(X,\dist)$ based at $w$}.  Gromov's argument on the existence of $(T,\varphi)$ in~\cite{Gr87} was revised by
Ghys and de la Harpe who provided a detailed proof of the theorem~\cite[Chapitre 2]{GH88}. The argument defines $(T,\varphi)$ via a pseudo-metric $\dist_T$ on the set $X$. From here on, we regard Gromov's approximating tree as the pseudo metric space $(X,\dist_T)$ and $\varphi$ as the identity map $(X,\dist)\to(X,\dist_T)$.   

By Remark~\ref{rem:Inverses},  the pseudo-metric $\dist_T$ is determined via its Gromov products with respect to $w$, which  we denote by  
\begin{equation}
\begin{split} (x|y)_w'& :=\frac12\left( \dist_T(x,w)+\dist_T(y,w)-\dist_T(x,y) \right)\\
&= \frac12\left(\dist(x,w)+\dist(y,w)-\dist_T(x,y)\right).
\end{split}\end{equation}
This expression yields
\begin{equation}\label{eq:10} \begin{split}\dist_T(x,y) & = (x|x)_w'+(y|y)_w'-2(x|y)_w'\\
& = (x|x)_w +(y|y)_w -2(x|y)_w' .
\end{split}
\end{equation}
In Gromov's argument~\cite[Page 156]{Gr87} the products $(x|y)'_w$ are given by the expression 
\begin{equation}\label{eq:11} (x|y)_w' = \sup_{\bar y \in S_{x,y}} \left\{ \min_k (y_k|y_{k+1})_w \right\}\end{equation}
where $S_{x,y}$ is the set of finite sequences of points $\bar y = (y_1,\ldots ,y_\ell)$ of $X$ such that $\ell\geq2$, $x=y_1$ and $y=y_\ell$.  

Expression~\eqref{eq:11} yields that Gromov's approximating tree is an optimal approximation in the following sense.

\begin{proposition}~\label{prop:Optimal} 
 Let $(X,\dist)$ be a finite pseudo-metric space and $(X,\dist_T)$ be its Gromov's approximating tree with respect to $w\in X$. If $(X,\dist_S)$ is a $0$-hyperbolic space such that $\dist_S(w,\cdot)=\dist(w,\cdot)$ and $\dist_S\leq \dist$, then $\dist_S\leq \dist_T$.    
\end{proposition}
\begin{proof}
Let us denote by $(\cdot|\cdot)_S$ the the Gromov product on $(X,\dist_S)$ with respect to $w$.
Since $(X,\dist_S)$ is $0$-hyperbolic, 
$(x|z)_S\geq\min\{(x|y)_S, (y|z)_S\}$ for any $x,y,z\in X$. It follows that 
\[(x|y)_S=\sup_{\bar y\in S_{x,y}} \left\{ \min_k(y_k|y_{k+1})_S \right\}. \]
Since $(x|y)\leq (x|y)_S$, expression~\eqref{eq:11} and the one above implies that $(x|y)' \leq (x|y)_S$ for any $x,y\in X$, which is equivalent to  $\dist_S\leq \dist_T$.    
\end{proof}

We are interested in computing the distance matrix $A$ of Gromov's approximating tree $(X,\dist_T)$, that is,
\begin{equation}\label{eq:12}
A=(a_{ij}),\qquad a_{ij}=\dist_T(x_i,x_j). 
\end{equation}
from the distance matrix $D$ of $(X,\dist)$. Let us name  this procedure.

\begin{definition} 
Let   \[\GTREE\colon \mathcal{D}_n \to \mathcal{D}_n\] be the function defined as follows. 
Let $(X,\dist)$ be a finite metric space with a fixed enumeration as~\eqref{eq:01} and a fixed point $w$ as in~\eqref{eq:04}, and let $D$ be its distance matrix as~\eqref{eq:03}. 
Then $\GTREE(D)$ is the   $n\times n$ distance matrix of Gromov's approximating tree $(X,\dist_T)$ with respect to $w$ as defined by~\eqref{eq:12}. 
\end{definition}
Now~\eqref{eq:12} can be written as  
\begin{equation}
    A=\GTREE(D).
\end{equation}
In the next section, it is explained how the expression~\eqref{eq:11} provides the connection between  computing Gromov's approximating tree  and computing the all pairs bottleneck paths (APBP) problem on an undirected graph. 

\section{Gromov's tree and the $\APBP$ problem}\label{sec:4}

In the following, all graphs considered are simple graphs; that is, they are undirected, no edge connects a vertex to itself, and no two vertices are connected by more than one edge. In the all pairs bottleneck paths (APBP) problem, there is a connected graph with real non-negative capacities on its edges. The problem asks to determine, for all pairs of  distinct vertices $s$ and $t$, the capacity of a single path for which a maximum amount of flow can be routed from $s$ to $t$. The maximum amount of flow of a path, also called the capacity of the path, is the minimum value of the capacities of its edges. It is known that the solution to the APBP problem in an undirected graph on $n$ vertices with real capacities can be computed in time $O(n^2)$, see~\cite{Hu61}.

It is enough to consider the APBP problem on complete graphs since zero is allowed as a capacity of an edge. Equivalently, if a connected graph with capacities is not complete, we can add the missing edges and assign them very small capacities, for example zero, to obtain a complete graph with the same solution to the APBP problem. 

\begin{definition}[$\mathcal{C}_n$]
Let $\mathcal{C}_n$ be the set of $n\times n$ symmetric matrices with non-negative real entries.    
\end{definition}

\begin{definition}[$\APBP$] 
 Let \[ \APBP\colon \mathcal{C}_n \to \mathcal{C}_n\] be the function such that
 $\APBP(C)$ is the matrix in $\mathcal{C}_n$ whose $(i,j)$-entry, for $i\neq j$, is the maximum amount of flow that can be routed from $i$ to $j$ in the complete graph with vertex set $\{1,\ldots ,n\}$ whose edge capacities are given by the entries of $C$. By convention, we assume that the diagonals of $C$ and $\APBP(C)$ coincide.  
\end{definition}

\begin{remark}\label{rem:apbpij}If $1\leq j<k\leq n$ then the  $(j,k)$-entry of $\APBP(C)$ is defined by 
\begin{equation}\label{eq:14}
    \APBP(C)_{jk}=\sup_{\bar y \in S_{j,k}} \left\{ \min_i \{ c_{y_i,y_{i+1}} \mid \bar y=(y_1,\ldots , y_\ell)\}  \right\}
\end{equation}
where $S_{j,k}$ is the set of finite sequences $\bar y = (y_1,\ldots , y_\ell)$ of elements of $\{1,\ldots ,n\}$ such that $\ell\geq2$, $i=y_1$ and $j=y_\ell$.
\end{remark}

Now we can describe the relation between the fuctions $\GTREE$ and $\APBP$. 
The Gromov products of Gromov's approximating tree of $(X,\dist)$ with respect to $w$ are defined by~\eqref{eq:11}. Recall that $L$ is the matrix of Gromov products of $(X,\dist)$ at $w$ and $A$ is the distance matrix of Gromov's approximating tree based at $w$.  
The main observation is that~\eqref{eq:11} is equivalent to the expression \begin{equation}
    \GPRD(A)=\APBP(L)
\end{equation}
in view of~\eqref{eq:14}. Since $\GPRD$ and $\IGPRD$ are inverses~\eqref{eq:08}, it follows that 
\begin{equation}
\begin{split}
    \GTREE(D)& =A =\IGPRD\circ\APBP(L). 
    \end{split}
\end{equation}
Since both $\APBP$ and $\IGPRD$ have time complexity $O(n^2)$, and   $L$ was defined as  $\GPRD(D)$, we have verified the following statement. 
 
\begin{proposition}\label{prop:7}
The function $\GTREE$ can be expressed as the composition
\begin{equation}\label{eq:17}
    \GTREE=\IGPRD\circ\APBP\circ \GPRD.
\end{equation}
In particular, $\GTREE$ can be computed on any $n\times n$ distance matrix with time complexity $O(n^2)$.
\end{proposition}
Equality~\eqref{eq:17} can be visualized as   
\begin{equation} \label{eq:BaseAlgorithm}
\begin{tikzcd}
   D \arrow[rr, bend left, "O(n^2)", "\text{$ \GPRD$}"']
   \arrow[rrrrrr, bend right, "\text{$\GTREE$}"'] && L \arrow[rr, bend left, "O(n^2)", "\text{$\APBP$}"' ] &&  M\arrow[rr, bend left, "O(n^2)", "\text{$\IGPRD$ }"'] && A.
\end{tikzcd}
\end{equation}

By the last statement, computing Gromov's approximating tree for a pointed metric space $(X,\dist,w)$ with $n+1$ points from its matrix of distances reduces to solving the $\APBP$ problem on a complete graph on $n+1$ vertices with edge capacities given by the matrix $L$ of Gromov products with respect to $w$. The row and the column of $L$ labelled by $w$ have only zero entries, therefore:

\begin{corollary}
Computing Gromov's approximating tree for a metric space with $n+1$ points reduces to solve the $\APBP$ problem on a complete graph with $n$ vertices. 
\end{corollary}

\section{Reducing $\APBP$ to $\GTREE$}

In this section we show that the $\APBP$ problem for a complete graph on $n$ vertices with capacities on its edges, can be reduced to computing Gromov's approximating tree for a pointed metric space with $n+1$ points from its distance matrix.

By~\eqref{eq:17}, we have that the restriction of $\APBP$ to  $\mathcal{P}_n$ satisfies
\begin{equation}\label{eq:18}
    \APBP|_{\mathcal{P}_n}=\GPRD\circ \GTREE \circ \IGPRD, 
\end{equation}
for every $n$.
Since $\mathcal{P}_n$ is a proper subset of $\mathcal{C}_n$, this expression does not reduce the computation $\APBP$ to a computation of $\GTREE$ over all of $\mathcal{C}_n$. However, any $C\in \mathcal{C}_n$ can be regarded as matrix in $\mathcal{P}_{n+1}$ modulo some minor adjustments, see Lemma~\ref{lem:metrization}, and then a minor variation of~\eqref{eq:18} holds over all $\mathcal{C}_n$, see~\eqref{eq:22}. 

\begin{definition}
    Let \[ \EXTE \colon \mathcal{C}_n \to \mathcal{C}_{n+1}\]
be the function defined as follows. If $C\in \mathcal{C}_n$ and $\mu=1+\max\limits_{ij} c_{ij}$ then 
\begin{equation}\label{eq:19}
 \EXTE (C)_{ij} = \begin{cases} 
      c_{ij} & 1\leq i,j\leq n \text{ and } i\neq j  \\
      \mu & i=j \text{ and } i\leq n \\
      0 & i=n+1 \text{ or } j=n+1 
   \end{cases} 
\end{equation}
\end{definition}

\begin{remark}
    For any $C\in\mathcal{C}_n$, the $n\times n$ matrix $\APBP(C)$ and $(n+1)\times (n+1)$ matrix $\APBP(\EXTE(C))$ coincide off the diagonal; specifically,
\begin{equation}\label{eq:20}
   \APBP(C)_{ij}= \APBP(\EXTE(C))_{ij} \quad \text{ for all $1\leq i,j\leq n$ such that $i\neq j$.}  
\end{equation}
\end{remark}

Now we state the key idea of this section.

\begin{lemma}\label{lem:metrization}
For any $C\in\mathcal{C}_n$, the matrix $\EXTE(C)$ belongs to $\mathcal{P}_{n+1}$. 
\end{lemma}

\begin{proof}
To fix notation, let $\EXTE(C)=(c_{ij}^*)$ so $c_{ij}^*$ is defined by~\eqref{eq:19}.
Let $D = (d_{ij})$ be the $(n+1) \times (n+1)$ matrix given by
    \[ d_{ij} := c^*_{ii} + c^*_{jj} - 2c^*_{ij}. \]
    Let $X = \{x_1, \ldots, x_n, x_{n+1}=w\}$ be a set with $n+1$ elements and fixed enumeration. Let $d:X^2 \rightarrow \mathbb{R}$ be given by $d(x_i,x_j) = d_{ij}$. We claim that $(X,d)$ is a metric space. 
    \begin{enumerate}
        \item Since $C$ is a symmetric matrix, we have that $d_{ij}=d_{ji}$ for all $1\leq i,j\leq n+1$.
        \item It follows from the definition that $d_{ii} = 0$ for all $i$. If $1\leq i <j\leq n$, then 
        \[ 2\mu \geq c^*_{ii} + c^*_{jj} - 2c^*_{ij} \ge 2\mu-2\max_{ij}c_{ij} = 2 > 0\]
        and hence 
        \[ 2\mu \geq d_{ij} \geq 2 > 0.\]
        If $1\leq i\leq n$ and $j=n+1$, then 
        \[d_{ij} = c^*_{ii} + c^*_{jj} - 2c^*_{ij} = \mu > 0.\]
        Therefore $d(x_i,x_j) = 0$ if and only if $i = j$.
        \item Let $i,j,k \in \{1,\ldots,n+1\}$. Now we  show that $d_{ij} \le d_{ik}+d_{kj}$. If any of $i,j,k$ are equal to each other, there is nothing to prove. Suppose $i,j,k$ are all distinct. If $k=n+1$ then the above  estimations show that $d_{ij}\leq 2\mu = d_{ik}+d_{kj}$. Analogously, if $i=n+1$ then $d_{ij}=\mu \leq \mu + d_{kj} = d_{ik}+d_{kj}$.  To conclude, let us assume that $i,j,k$ are all distinct and less than $n+1$. Then we have
        \[ d_{ij} \le d_{ik} + d_{kj} \quad \iff \quad 2c^*_{kk} \ge c^*_{ik} + c^*_{jk} - c^*_{ij}, \]
        but we see that
        \[ 2c^*_{kk} = 2\mu \geq  c_{ik}+c_{jk} -c_{ij} =c^*_{ik}+c^*_{jk} - c^*_{ij}. \]
        So, $d$ satisfies the triangle inequality.
    \end{enumerate}
    Therefore, $(X,d)$ is a metric space, in particular, a pseudo-metric space. Observe that  for any $i,j < n
    +1$, 
    \[ c^*_{ij} = \frac{1}{2}(c^*_{ii}+c^*_{jj}-d_{ij}) = \frac{1}{2}(d_{i,n+1}+d_{j,n+1}-d_{ij}) = (x_i|x_j)_{w}. \]
    Consequently, $\EXTE(C) = \GPRD(D)$ and therefore $\EXTE(C) \in \mathcal{P}_{n+1}$.   
\end{proof}

By Lemma~\ref{lem:metrization}, we can regard $\EXTE$ as a function with domain $\mathcal{C}_n$ and codomain $\mathcal{P}_{n+1}$, that is,
\begin{equation}\label{eq:21}
\EXTE\colon \mathcal{C}_n \to \mathcal{P}_{n+1},
\end{equation}
In view of~\eqref{eq:21}, the statement of~\eqref{eq:18} implies
\begin{equation}\label{eq:22} \APBP\circ \EXTE = \GPRD\circ \GTREE \circ \IGPRD\circ \EXTE
\end{equation}
over $\mathcal{C}_n$ for every $n$. Putting together~\eqref{eq:20} and~\eqref{eq:22} we obtain that the computation of $\APBP$ can be reduced to a computation of $\GTREE$ in the following sense.
\begin{proposition}\label{prop:11}
 For any matrix $C\in\mathcal{C}_n$, 
 \begin{equation}\label{eq:23}
 \APBP(C)_{ij}=\GPRD\circ \GTREE \circ \IGPRD\circ \EXTE (C)_{ij}
\end{equation}
for every $1\leq i,j\leq n$ with $i\neq j$.
\end{proposition}


 This last statement implies that we can compute $\APBP(C)$ via finding the Gromov's approximating tree of a metric space whose Gromov's products are given by $\EXTE (C)$. Note that $\IGPRD\circ\EXTE (C)$ is a distance matrix of a metric space with $n+1$ points. Hence, Propositions~\ref{prop:7} and~\ref{prop:11} yield: 
 
 \begin{corollary}\label{cor:14}
     Solving the APBP problem on an undirected graph on $n$ vertices with positive capacities is equivalent to finding Gromov's approximating tree for a metric space on $n+1$ points.
 \end{corollary}

\section{Gromov's Approximating Tree for Graphs}

In this section we  describe how to compute Gromov's approximating tree from the adjacency matrix of a graph in quadratic time. 
More specifically, we consider the case that $(X,\dist)$ is a \emph{graph metric space}, that means, the metric $\dist(x,y)$ can be realized as the shortest edge-path between vertices of an undirected, connected graph $\Gamma$ with vertex set $X$, and edges of  length one. We aim to compute the Gromov's approximating tree of $(X,\dist)$ from the adjacency matrix of $\Gamma$.  

Let $G$ be the adjacency matrix $\Gamma$,  
\begin{equation}\label{eq:Gmatrix} G= (g_{ij}),\qquad g_{ij}= \begin{dcases} 
      1 & \text{if $\dist(x_i , x_j)=1$},  \\
      0 & \text{if $\dist(x_i , x_j)\neq 1$}. 
   \end{dcases}
\end{equation}
As in the previous sections, let $D$ denote the distance matrix of $(X,\dist, w)$ and let $A$ be the distance matrix of Gromov's approximating tree $(X,\dist_T)$. In this section we address the problem to compute the matrix
\[ A=\GTREE(D) \] 
from the adjacency matrix $G$. 

Observe that one can compute $G$ from $D$ in  time $O(n^2)$; however, the computation of $D$ from $G$ is not known to be quadratic. It can be done in time $O(n^\omega \log n)$, where $O(n^\omega)$ is the complexity of multiplying two $n\times n$ matrices (see the work of Seidel~\cite{SE95}). The best known upper bound of $\omega$ is $2.371552$, see the recent article~\cite{williams2023new}. Hence, based on~\eqref{eq:BaseAlgorithm}, an algorithm that takes as an input the adjacency matrix $G$ of the graph and a distinguish vertex $w$,  and outputs the distance matrix $A$ of the Gromov's approximating tree can be described by  the diagram 
\[ 
\begin{tikzcd}
   G \arrow[rr, bend left, "O(n^\omega \log n)", "\text{Seidel}"']  && 
   D \arrow[rr, bend left, "O(n^2)", "\text{$\GPRD$}"'] && L \arrow[rr, bend left, "O(n^2)", "\text{APBP}"' ] &&  M\arrow[rr, bend left, "O(n^2)", "\text{$\IGPRD$}"'] && A,
\end{tikzcd}
\]
where
\[ L=\GPRD(D),\  M=\APBP(L),\ \text{ and}\   A=\IGPRD(M).\] 
This procedure can be done in time  $O(n^\omega \log n)$ and it is not the best way to proceed. Below we describe a quadratic algorithm that takes $G$ as input and outputs $A$, via solving the APBP problem and bypassing the computation of $D$ and $L$.  

Let $K=\LPRD(G)$ denote 
the $n\times n$ matrix 
containing the Gromov products of pairs of points that are adjacent, that is, 
\begin{equation}
   \LPRD(G)=(k_{ij}),\qquad k_{ij}= \left\{
  \begin{array}{lr} 
      \ell_{ij} & \dist (x_i,x_j)\leq 1, \\
      0 & \text{otherwise}.
      \end{array}
\right.
\end{equation}
Let us show that one can compute $\LPRD(G)$  in time $O(n^2)$.  By definition of $\ell_{ij}$, see~\eqref{eq:relations}, we have that
\[ k_{ii}=\ell_{ii}=\dist(w,x_i).\]
Hence the values  $k_{ii}$ can be computed using Dijkstra's algorithm~\cite{Dij59} which has time complexity $O(n^2)$ when using the adjacency matrix $G$. On the other hand, if $i\neq j$ and $\dist(x_i,x_j)=1$, we have that
\begin{equation}
    k_{ij}=\ell_{ij}=\frac{1}{2}(k_{ii}+k_{jj}-1)
\end{equation}
in view of~\eqref{eq:05}.  Therefore we can compute $k_{ij}$ for $i\neq j$ by checking every entry of the adjacency matrix $G$ and using the previous formula; this procedure clearly has time complexity $O(n^2)$. In a diagram, 
\[ 
\begin{tikzcd}
   G\arrow[rr, bend left, "O(n^2)"] && K,
\end{tikzcd}  \qquad \qquad K=\LPRD(G)
\]
based on applying Dijkstra's algorithm to $G$ and then traversing the matrix $G$.  
The key point of this section is the following statement.
\begin{lemma}\label{lem:incomplete-apbp} \label{lem:incomplete-apbp2} 
$   
  \APBP(K)=\APBP(L).
$
\end{lemma}
\begin{proof}
We regard a path in $\Gamma$ from a vertex $x$ to a vertex $y$ as a finite sequence of vertices $(z_1, z_2, \ldots, z_t)$ such that $x=z_1$, $y=z_t$, and $z_i,z_{i+1}$ are adjacent vertices in $\Gamma$ for all $i$. The length of such a path is $t-1$, and a geodesic path between $x$ and $y$ is a path of minimal length.

Let $x_i, x_j \in X$ with $i \neq j$. Denote by $S_{ij}$ the set of all {sequences} $\bar y=(y_1, y_2, \ldots, y_s)$ with $y_1 = x_i$ and $y_s = x_j$, and denote by $P_{ij}$ the set of all \emph{paths} $\bar z=(z_1, z_2, \ldots, z_t)$ in $\Gamma$ with $z_1 = x_i$ and $z_t = x_j$. Define 
    \[\alpha_{ij} = \sup\limits_{S_{ij}}\left\{ \min_{k} (y_{k-1}|y_{k})_w \right\}, \quad \beta_{ij} = \sup\limits_{P_{ij}}\left\{ \min_{k} (z_{k-1}|z_{k})_w \right\}.\]
    Note that $\alpha_{ij}$ and $\beta_{ij}$ are the maximal capacities from $x_i$ to $x_j$ on the graphs with edge capacities determined by $L$ and $K$, respectively.
    By definition $\ell_{ii}=k_{ii}$ for all $i$, hence the proposition follows by verifying that $\alpha_{ij} = \beta_{ij}$. Since $P_{ij} \subseteq S_{ij}$ it follows that  $\alpha_{ij} \ge \beta_{ij}$.  

Let  $(y_1,y_2,\ldots,y_s) \in S_{ij}$. Let $(z_1,z_2, \ldots,  z_t)\in P_{ij}$ be a path in $\Gamma$ obtained from concatenating shortest paths from $y_k$ to $y_{k+1}$ for all $k \le s$. 
    Let $p^{\ast}$, $(p+1)^{\ast}$ be the integers such that the chosen geodesic between $y_p$ and $y_{p+1}$ consists of the sequence of $z_q$'s such that $p^{\ast} \le q \le (p+1)^{\ast}$. To show that $\alpha_{ij}\leq\beta_{ij}$, it is enough to verify
    \[ (y_p|y_{p+1})_w \le \min_{p^{\ast} \le q < (p+1)^{\ast}}(z_q|z_{q+1}). \]
    For any $p^{\ast} \le q < (p+1)^{\ast}$, since $z_q$ and $z_{q+1}$ are consecutive vertices on a geodesic path from $y_p$ to $y_{p+1}$,
    \[\dist(y_p,y_{p+1})= \dist(y_p,z_q)+\dist(z_{q+1},y_{p+1}) + 1    .\]
    It follows from  the triangle inequality and the previous equality that
    \begin{align*}
        & 2(y_{p}|y_{p+1})_w =  \dist(w,y_p) +\dist(w,y_{p+1} )-\dist(y_p,y_{p+1}) \\
        &\le   \dist(w,z_q)+\dist(z_q,y_p) +\dist(w,z_{q+1} ) +\dist(z_{q+1},y_{p+1} ) -\dist(y_p,y_{p+1}) \\
        &=   \dist(w,z_q) +\dist(w,z_{q+1}) -1   \\
        &=  2(z_q|z_{q+1})_w,
    \end{align*}
    thus proving the desired inequality. Therefore $\alpha_{ij} = \beta_{ij}$, concluding the proof.
\end{proof}

Putting together~\eqref{eq:BaseAlgorithm} and Lemma~\ref{lem:incomplete-apbp2}, we obtain an algorithm, described by the following diagram, that computes from the adjacency matrix $G$ the distance matrix $A$ of Gromov approximating tree: 
\begin{equation} \label{eq:BaseAlgorithm2}
\begin{tikzcd}
   G \arrow[rr, bend left, "O(n^2)", "\text{$\LPRD$}"'] && K \arrow[rr, bend left, "O(n^2)", "\text{APBP}"' ] &&  M\arrow[rr, bend left, "O(n^2)", "\text{$\IGPRD$}"'] && A.
\end{tikzcd}
\end{equation}
where $M=\APBP(L)$. 
Equivalently, 
from~\eqref{eq:17}, it follows that
\[ \begin{split} 
A=\GTREE(D)  &  = 
\IGPRD(\APBP(\LPRD(G))),  
\end{split}\]
where the last expression  provides an algorithm takes as an input the adjacency matrix of $G$ and outputs the distance matrix $A$ of the Gromov approximating tree. Since $\IGPRD$, $\APBP$ and $\LPRD$ can be computed in time $O(n^2)$, we have verified that there is an algorithm that computes Gromov's approximating tree of a connected graph on $n$ vertices from its adjacency matrix in time $O(n^2)$. This proves Proposition~\ref{prop:GraphsIntro} stated in the introduction.

\bibliographystyle{alpha} \bibliography{refs}

\end{document}